\documentclass[11pt,a4paper]{article}

\usepackage{a4wide}
\usepackage[english]{babel}
\usepackage{graphicx}
\usepackage{amsmath}
\usepackage{amssymb}
\usepackage{xspace}
\usepackage{xcolor}
\usepackage{microtype}%if unwanted, comment out or use option "draft"
\usepackage{algorithm}
\usepackage[noend]{algpseudocode}

\usepackage{xspace}
 
\usepackage[mathlines]{lineno}
\graphicspath{{./images/}}%helpful if your graphic files are in another directory

\newtheorem{defin}{Definition}
 
\newtheorem{theo}[defin]{Theorem}
 \newenvironment{theorem}{\begin{theo} \sl}{\end{theo}}
\newtheorem{lem}[defin]{Lemma}
 \newenvironment{lemma}{\begin{lem} \sl}{\end{lem}}
\newtheorem{coro}[defin]{Corollary}
 \newenvironment{corollary}{\begin{coro} \sl}{\end{coro}}
\newtheorem{prop}[defin]{Proposition}

\newenvironment{proof}{\emph{Proof.}}{\hfill $\Box$\\}

\newcommand{\Vis}{\mathord{\it Vis}}

\newcommand{\etal}{\emph{et~al.}\xspace}

\title{Constrained Routing Between Non-Visible Vertices\thanks{An extended abstract of this paper appeared in the proceedings of the 23rd Annual International Computing and Combinatorics Conference (COCOON 2017)~\cite{BKRV2017Routing}. 
P.~B. is supported in part by NSERC. 
M.~K.~was partially supported by MEXT KAKENHI Nos.~12H00855, 15H02665, and 17K12635. 
A.~v.~R. was supported by JST ERATO Grant Number JPMJER1201, Japan. 
S.~V.~was supported in part by NSERC and the Carleton-Fields Postdoctoral Award.}}

\author{Prosenjit Bose \and Matias Korman \and Andr\'e van Renssen  \and Sander Verdonschot}

\date{}

\begin{document}

\maketitle

\begin{abstract}
In this paper we study local routing strategies on geometric graphs. Such strategies use geometric properties of the graph like the coordinates of the current and target nodes to route. Specifically, we study routing strategies in the presence of constraints which are obstacles that edges of the graph are not allowed to cross. Let $P$ be a set of $n$ points in the plane and let $S$ be a set of line segments whose endpoints are in $P$, with no two line segments intersecting properly. We present the first deterministic 1-local $O(1)$-memory routing algorithm that is guaranteed to find a path between two vertices in the \emph{visibility graph} of $P$ with respect to a set of constraints $S$. The strategy never looks beyond the direct neighbors of the current node and does not store more than $O(1)$-information to reach the target. 

We then turn our attention to finding competitive routing strategies. We show that when routing on any triangulation $T$ of $P$ such that $S\subseteq T$, no $o(n)$-competitive routing algorithm exists when the routing strategy restricts its attention to the triangles intersected by the line segment from the source to the target (a technique commonly used in the unconstrained setting). Finally, we provide an $O(n)$-competitive deterministic 1-local $O(1)$-memory routing algorithm on any such $T$, which is optimal in the worst case, given the lower bound.  
\end{abstract}

\section{Introduction}

A distributed routing strategy is an algorithm that, given an intended destination, determines to which of the neighbors of the current vertex to forward the message. A distributed routing strategy is {\em local} when that decision is based solely on knowledge of the location of the current vertex, the location of its neighbors and a constant amount of additional information (such as the location of the source vertex and destination vertex). A traditional approach to this routing problem is to build a {\em routing table} at each node, explicitly storing for each destination vertex, which neighbor of the current vertex to send the message. 

In this paper, we study routing algorithms on geometric graphs and try to circumvent the use of routing tables by leveraging geometric information. This model excels when global network information is either difficult to compute or unavailable. This can be due to efficiency or time constraints, or the network could be dynamic and updating this information might be a problem. To tackle this additional hurdle, we use geometric tools. A routing algorithm is considered \emph{geometric} when the graph that is routed on is embedded in the plane, with edges being straight line segments connecting pairs of vertices. Edges are usually weighted by the Euclidean distance between their endpoints. Geometric routing algorithms are particularly useful in wireless sensor networks \mbox{(see \cite{G09} and \cite{R09}} for surveys on the topic), since nodes often connect only to nearby nodes. Thus, by exploiting geometric properties (such as distance, or the coordinates of the vertices) we can devise algorithms to guide the search and remove the need for routing tables. 

We consider the following setting: let $P$ be a set of $n$ points in the plane and let $S$ be a set of line segments whose endpoints are in $P$, with no two line segments of $S$ properly intersecting (i.e., intersections only occur at endpoints). Two vertices $u$ and $v$ are \textit{visible} if and only if either the line segment $uv$ does not properly intersect any constraint or the segment $uv$ is itself a constraint. If two vertices $u$ and $v$ are visible, then the line segment $uv$ is a \emph{visibility edge}. The \emph{visibility graph} of $P$ with respect to a set of constraints $S$, denoted $\Vis(P,S)$, has $P$ as vertex set and all visibility edges as edge set. In other words, it is the complete graph on $P$ minus all edges that properly intersect one or more constraints in~$S$. 

This model has been studied extensively in the context of motion planning. Clarkson \cite{C87} was one of the first to study this problem. He showed how to construct a $(1+\epsilon)$-{\em spanner} of $\Vis(P,S)$ with a linear number of edges. A subgraph $H$ of $G$ is called a $t$-{\em spanner} of $G$ (for $t\geq 1$) if for each pair of vertices $u$ and $v$, the shortest path in $H$ between $u$ and $v$ has length at most $t$ times the shortest path between $u$ and $v$ in $G$. The smallest value $t$ for which $H$ is a $t$-spanner is the \emph{spanning ratio} or \emph{stretch factor} of $H$. Following Clarkson's result, Das \cite{D97} showed how to construct a spanner of $\Vis(P,S)$ with constant spanning ratio and constant degree. Bose and Keil \cite{BK06} showed that the Constrained Delaunay Triangulation is a 2.42-spanner of $\Vis(P,S)$. Recently, the constrained half-$\Theta_6$-graph (which is identical to the constrained Delaunay graph whose empty visible region is an equilateral triangle) was shown to be a plane 2-spanner of $\Vis(P,S)$~\cite{BFRV12Constrained} and all constrained $\Theta$-graphs with at least 6 cones were shown to be spanners as well~\cite{BR14}. 

Spanners of $\Vis(P, S)$ are desirable because they can be sparse and the bounded stretch factor certifies that paths do not make large detours compared to the shortest path in $\Vis(P, S)$. Thus, by using a spanner we can compact a potentially large network using a small number of edges at the cost of a small detour when sending the messages. Unfortunately, little is known on how to route once the network has been built.
Bose~\etal~\cite{BFRV2017RoutingJournal} showed that it is possible to route locally and 2-competitively between any two visible vertices in the constrained $\Theta_6$-graph. Additionally, an 18-competitive routing algorithm between any two visible vertices in the constrained half-$\Theta_6$-graph was provided (the definition of these two graphs as well as formal definitions of {\em local} and {\em competitiveness ratio} are given in Section~\ref{sec_thetadef}). Intuitively, a routing strategy is $c$-competitive if, for every pair of vertices $u$ and $v$, the length of its path between $u$ and $v$ is at most $c$ times that of the shortest path between the vertices. While it seems like a serious shortcoming that these routing algorithms only route between pairs of visible vertices, in the same paper the authors also showed that no deterministic local routing algorithm can be $o(\sqrt{n})$-competitive between all pairs of vertices of the constrained $\Theta_6$-graph, regardless of the amount of memory one is allowed to use. As such, the best one can hope for in this setting is an $O(\sqrt{n})$ competitive routing ratio.

In this paper, we develop routing algorithms that work between any pair of vertices in the constrained setting. This is, to the best of our knowledge, the only deterministic 1-local routing strategy that works for vertices that cannot see each other in the constrained setting. We provide a non-competitive 1-local routing algorithm on the visibility graph of $P$ with respect to a set of constraints $S$. Our algorithm locally computes a sparse subgraph of the visibility graph and routes on it. 

Parallel to this work, we designed a routing strategy that specifically works in the visibility graph directly (without having to compute a subgraph). The details of this routing strategy are quite lengthy, so they are given in a companion paper~\cite{bkrv-rvg-17}. Like Theorem~\ref{theo_routing1} presented in this paper, that algorithm is 1-local and non-competitive. In a nutshell, the algorithm presented in this paper has two stages: (i) compute a subset of edges to remove so that the graph becomes planar and then (ii) use a previously known routing technique on the subgraph. The companion paper avoids the 2-stage approach. Instead, it gives a (more involved) routing strategy tailored specifically for the visibility graph. Because of the simplicity of the approach presented in this paper, we believe that the results presented here can more easily generalize to other families of graphs.

After presenting this routing strategy, we also show that when routing on any triangulation $T$ of $P$ such that $S\subseteq T$, no $o(n)$-competitive routing algorithm exists when only considering the triangles intersected by the line segment from the source to the target. Considering only this subset of triangles is a technique commonly used to route in the unconstrained setting. Finally, we provide an $O(n)$-competitive 1-local routing algorithm on $T$, which is optimal in the worst case, given the lower bound.

\section{Preliminaries}\label{sec_thetadef}
\subsection{Routing Model}
Given a graph $G=(V,E)$, the $k$-neighborhood of a vertex $u \in V$ is the set of vertices in the graph that can be reached from $u$ by following at most $k$ edges (and is denoted by $N_k(u)$). We assume that the only information stored at each vertex of the graph is $N_k(u)$ for some fixed constant $k$. Since our graphs are geometric, vertices are points in the plane. We label each vertex by its coordinates in the plane.

We are interested in deterministic $k$-local, $m$-memory {\em routing algorithms}. That is, the vertex to which the message is forwarded is determined by a deterministic function that only depends on $s$ (the source vertex), $u$ (the current vertex), $t$ (the destination vertex), $N_k(u)$ and a string $M$ of at most $m$ words. This string $M$ is stored within the message and can be modified before forwarding the message to the next node. For our purposes, we consider a word (or unit of memory) to consist of a $\log_2 n$ bit integer or a point in $\mathbb{R}^2$. 

We focus on algorithms that guarantee that the message will arrive at its destination (i.e., for any graph $G$ and source vertex $s$, by repeatedly applying the routing strategy we will reach the destination vertex in a finite number of steps). We will focus on the case where $k=1$ and $|M| \in O(1)$. Thus, for brevity, by {\em local} routing algorithm we mean that the algorithm is 1-local, uses a constant amount of memory, and arrival of the message at the destination is guaranteed.

\subsection{Competitiveness}
Intuitively speaking, we can evaluate how good a routing algorithm is by looking at the detour it makes (i.e., how long are the paths compared to the shortest possible). We say that a routing algorithm is {\em $c$-competitive} with respect to a graph $G$ if, for any pair of vertices $u,v\in V$, the total distance traveled by the message is not more than $c$ times the shortest path length between $u$ and $v$ in $G$. The \emph{routing ratio} of an algorithm is the smallest $c$ for which it is $c$-competitive.

\subsection{Graph Definitions}
In this section we introduce the $\Theta_m$-graph, a graph that plays an important role in our routing strategy. We begin by defining this graph and some known variations. Define a \emph{cone} $C$ to be the region in the plane between two rays originating from a vertex (the vertex itself is referred to as the {\em apex} of the cone). When constructing a (constrained) $\Theta_m$-graph of a set $P$ of $n$ vertices we proceed as follows: for each vertex $u\in P$ consider $m$ rays originating from $u$ so that the angle between two consecutive rays is $2 \pi / m$. Each pair of consecutive rays defines a cone. We orient the rays in a way that the bisector of one of the cones is the vertical halfline through $u$ that lies above $u$. Let this cone be $C_0$ of $u$. We number the other cones $C_1, \ldots, C_{m-1}$ in clockwise order around $u$ (see Fig.~\ref{fig:Cones}). We apply the same partition and numbering for the other vertices of $P$. We write $C_i^u$ to indicate the $i$-th cone of a vertex $u$, or $C_i$ if $u$ is clear from the context. For ease of exposition, we only consider point sets in general position: no two vertices lie on a line parallel to one of the rays that define the cones, no two vertices lie on a line perpendicular to the bisector of a cone, no three vertices are collinear, and no four vertices lie on the boundary of any circle. All these assumptions can be removed using classic symbolic perturbation techniques~\cite{EC95,EM90,Y90}. 

The $\Theta_m$-graph is constructed by adding an edge from $u$ to the {\em closest} vertex in each cone $C_i$ of each vertex $u$, where distance is measured along the bisector of the cone. More formally, we add an edge between two vertices $u$ and $v\in C_i^u$ if for all vertices $w \in C_i^u$ it holds that $|u v'| \leq |u w'|$ (where $v'$ and $w'$ denote the projection of $v$ and $w$ on the bisector of $C_i^u$ and $|x y|$ denotes the length of the line segment between two points $x$ and $y$). Note that our general position assumption implies that each vertex adds at most one edge per cone.

\begin{figure}[ht]
  \begin{minipage}[t]{0.45\linewidth}
    \begin{center}
      \includegraphics{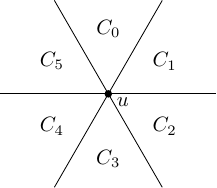}
    \end{center}
    \caption{Vertex $u$ and the six cones that are generated in the $\Theta_6$-graph. All vertices of $P$ have a similar construction with six cones and the same orientation.}
    \label{fig:Cones}
  \end{minipage}
  \hspace{0.05\linewidth}
  \begin{minipage}[t]{0.45\linewidth}
    \begin{center}
      \includegraphics{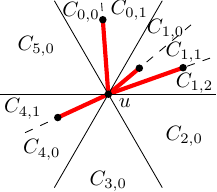}
    \end{center}
    \caption{When $u$ is the endpoint of one or more constraints (denoted as red thick segments in the figure), some cones may be partitioned into subcones.}
    \label{fig:ConstrainedCones}
  \end{minipage}
\end{figure}

The $\Theta_m$-graph has been adapted to the case where constraints are present; for every constraint whose endpoint is $u$, consider the ray from $u$ to the other endpoint of the constraint. These rays split the cones into several \emph{subcones} (see Fig.~\ref{fig:ConstrainedCones}). We use $C_{i, j}^u$ to denote the $j$-th subcone of $C_i^u$ (also numbered in clockwise order). Note that if some cone $C_i$ is not subdivided with this process, we simply have $C_i=C_{i,0}$ (i.e., $C_i$ is a single subcone). Further note that we treat the subcones as closed sets (i.e., contain their boundary). Thus, when a constraint $c = (u, v)$ splits a cone of $u$ into two subcones, vertex $v$ lies in both subcones. Due to the general position assumption, this is the only case where a vertex can be in two subcones of $u$.

With the subcone partition we can define the {\em constrained} $\Theta_m$-graph: for each subcone $C_{i, j}$ of each vertex $u$, add an edge from $u$ to the {\em closest} vertex that is in that subcone and can see $u$ (if any exist). Note that distance is measured along the bisector of the original cone (\emph{not the subcone}, see Fig.~\ref{fig:Projection}). More formally, we add an edge between two vertices $u$ and $v\in C_{i, j}^u$ if $v$ can see $u$, and for all vertices $w \in C_{i, j}^u$ that can see $u$ it holds that $|u v'| \leq |u w'|$ (where $v'$ and $w'$ denote the projection of $v$ and $w$ on the bisector of $C_i^u$ and $|x y|$ denotes the length of the line segment between two points $x$ and $y$). Note that our general position assumption implies that each vertex adds at most one edge per subcone.

\begin{figure}[ht]
  \begin{center}
    \includegraphics{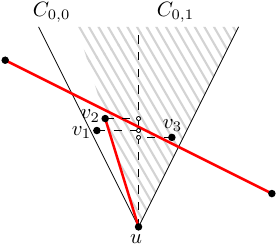}
  \end{center}
  \caption{The constraint $(u,v_2)$ partitions $C_0^u$ into two subcones. Subcone $C_{0,0}$ contains two visible vertices, out of which $v_1$ is closest to $u$. Subcone $C_{0,1}$ only contains one visible vertex: $v_2$ (note that $v_3$ is closer to $u$ than $v_2$, but it is not visible).} 
  \label{fig:Projection}
\end{figure}

Although constrained $\Theta_m$-graphs are quite sparse, sometimes it is useful to have even fewer edges. Thus, we introduce the constrained  {\em half}-$\Theta_6$-graph. This is the natural generalization of the half-$\Theta_6$-graph as described by Bonichon \etal\cite{BGHI10}, who considered the case where no constraints are present. This graph is defined for any even $m$, but in this paper we will consider only the case where $m=6$. Thus, for simplicity in notation we define only the constrained half-$\Theta_6$-graph.

The main change with respect to the constrained $\Theta_6$-graph is that edges are added only in every second cone. More formally, we rename the cones of a vertex $u$ to $(C_0, \overline{C_1}, C_2, \overline{C_0}, C_1, \overline{C_2})$ (as usual, we use clockwise order starting from the cone containing the positive $y$-axis). The cones $C_0$, $C_1$, and $C_2$ are called \emph{positive} cones and $\overline{C_0}$, $\overline{C_1}$, and $\overline{C_2}$ are called \emph{negative} cones. 

We use $C^u_i$ and $\overline{C^u_i}$ to denote cones $C_i$ and $\overline{C_i}$ with apex $u$. Note that, by the way the cones are labeled, for any two vertices $u$ and $v$, it holds that $v \in C^u_i$ if and only if $u \in \overline{C^v_i}$. Analogous to the subcones defined for the constrained $\Theta_6$-graph, constraints split cones into subcones. We call a subcone of a positive cone a positive subcone and a subcone of a negative cone a negative subcone (see Fig.~\ref{fig:ConstrainedConesHalfGraph}). 

\begin{figure}[ht]
  \begin{center}
    \includegraphics{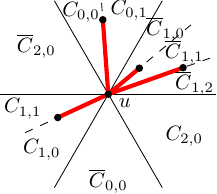}
  \end{center}
 \caption{The constrained half-$\Theta_6$-graph uses a construction similar to that of Fig.~\ref{fig:ConstrainedCones}. Notice that we have the same number of cones, but different notation is used.}
 \label{fig:ConstrainedConesHalfGraph}
\end{figure}

In the constrained half-$\Theta_6$-graph we add edges like in the constrained-$\Theta_6$-graph, but only in the positive cones (and their subcones). Hence, the constrained half-$\Theta_6$-graph is a proper subgraph of the constrained-$\Theta_6$-graph. We look at the undirected version of these graphs, i.e. when an edge is added, both vertices are allowed to use it. This is consistent with previous work on $\Theta$-graphs. 

Finally, we define the {\em constrained Delaunay triangulation}~\cite{BK06}. Given any two visible vertices $p$ and $q$, the constrained Delaunay triangulation contains an edge between $p$ and $q$ if and only if $p q$ is a constraint or there exists a circle $O$ with $p$ and $q$ on its boundary such that there is no vertex of $P$ in the interior of $O$ that is visible to both $p$ and $q$.

\section{Local Routing on the Visibility Graph}
\label{sec:routing}
In the unconstrained setting there is a very simple local routing algorithm for $\Theta_m$-graphs. The algorithm (often called $\Theta$-routing) greedily follows the edge to the closest vertex in the cone that contains the destination. This strategy is guaranteed to work for $m\geq 4$, and is competitive when $m\geq 7$~\cite{C87}.

This strategy does not easily extend to the case where constraints are present: it is possible that the cone containing the destination does not have any visible vertices, since a constraint blocks its visibility (see Fig.~\ref{fig:ThetaRoutingStuck}). Having no edge in that cone, it is unclear how to reach the destination that lies beyond the constraint. In fact, given a set $P$ of vertices in the plane and a set $S$ of disjoint segments, no  deterministic local routing algorithm is known for routing on $\Vis(P,S)$ that guarantees delivery of the message.

\begin{figure}[ht]
  \begin{center}
    \includegraphics{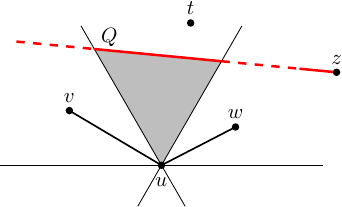}
  \end{center}
  \caption{The classic $\Theta$-routing algorithm can get stuck in the presence of constraints. In the example, $u$ does not have any edge in the cone that contains the destination $t$, because it is behind a constraint.}
  \label{fig:ThetaRoutingStuck}
\end{figure}

When the destination $t$ is visible to the source $s$, it is possible to route locally by essentially ``following the line segment $st$'', since no constraint can intersect $st$. This approach was used to give a 2-competitive 1-local routing algorithm on the constrained half-$\Theta_6$-graph, provided that $t$ is in a positive cone of $s$~\cite{BFRV2017RoutingJournal}. In the case where $t$ is in a negative cone of $s$, the algorithm is much more involved and the competitive ratio jumps to 18.

The stumbling block of all known approaches is the presence of constraints. In a nutshell, the problem is to determine how to ``go around'' a constraint in such a way as to reach the destination and prevent cycling. This gives rise to the following question: does there exist a deterministic 1-local routing algorithm that always reaches the destination when routing on the visibility graph? In this section, we answer this question in the affirmative. We provide a 1-local algorithm that is guaranteed to route from a given source to a destination, in the presence of constraints. The main idea is to route on a planar subgraph of $\Vis(P,S)$ that can be computed locally.

In \cite{KSU99} it was shown how to route locally on a plane geometric graph. Subsequently, in \cite{BMSU01}, a modified algorithm was presented that seemed to work better in practice. Both algorithms are described in detail in \cite{BMSU01}, where the latter algorithm is called FACE-2 and the former is called FACE-1. Neither of the algorithms is competitive. FACE-1 reaches the destination after traversing at most $\Theta(n)$ edges in the worst case and FACE-2 traverses $\Theta(n^2)$ edges in the worst case. Although FACE-1 performs better in the worst case, FACE-2 performs better on average in random graphs generated by vertices uniformly distributed in the unit square. 

Coming back to our problem of routing locally from a source $s$ to a destination $t$ in $\Vis(P,S)$, the main difficulty for using the above strategies is that the visibility graph is not a plane graph. Its seems counter-intuitive that having more edges makes the problem of finding a path more difficult. Indeed, almost all local routing algorithms in the literature that guarantee delivery do so by routing on a plane subgraph that is computed locally. For example, in \cite{BMSU01}, a local routing algorithm is presented for routing on a unit disk graph and the algorithm actually routes on a planar subgraph known as the Gabriel graph. However, none of these algorithms guarantee delivery in the presence of constraints. In this section, we adapt the approach from \cite{BMSU01} by showing how to locally identify the edges of a planar spanning subgraph of $\Vis(P,S)$, which then allows us to use FACE-1 or FACE-2 to route locally on $\Vis(P,S)$.

Our aim is to route on the constrained half-$\Theta_6$-graph. This graph was shown to be a plane 2-spanner of $\Vis(P,S)$~\cite{BFRV12Constrained}. The authors also showed a partial routing result (only between visible vertices) on this graph~\cite{BFRV2017RoutingJournal}.

\begin{lemma}{\emph{(Lemma~1 of~\cite{BFRV12Constrained}, see Figure~\ref{fig:VisiblePointInsideTriangle})}}
  \label{lem:ConvexChain}
  Let $u$, $v$, and $w$ be three points in the plane such that $u w$ and $v w$ are visibility edges and $w$ is not the endpoint of a constraint intersecting the interior of triangle $u v w$. Then there exists a convex chain of visibility edges from $u$ to $v$ in triangle $u v w$, such that the polygon defined by $u w$, $w v$ and the convex chain does not contain any constraint or vertex of $P$.
\end{lemma}

 \begin{figure}[ht]
  \begin{center}
    \includegraphics{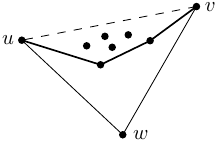}
  \end{center}
  \caption{The (simplified) figure from~\cite{BFRV12Constrained} illustrating Lemma~\ref{lem:ConvexChain}.}
  \label{fig:VisiblePointInsideTriangle}
\end{figure}

We now show how to locally identify the edges of the constrained half-$\Theta_6$-graph and distinguish them from other edges of $\Vis(P,S)$.

\begin{lemma} \label{lem:localnegedge}
 Let $u$ and $v$ be visible vertices such that $u \in \overline{C_0^v}$. Then $uv$ is an edge of the constrained half-$\Theta_6$-graph if and only if $v$ is the vertex whose projection on the bisector of $C_0^u$ is closest to $u$, among all vertices in $C_0^u$ visible to $v$ and not blocked from $u$ by constraints incident on $v$.
\end{lemma}
\begin{proof}
We will prove the claim by contradiction. First, suppose that $v$ is not closest to $u$ among the vertices in $C_0^u$ visible to $v$ and not blocked by constraints incident on $v$ (see Fig.~\ref{fig:LocalConstruction}a). Then there are one or more vertices whose projection on the bisector is closer to $u$. Among those vertices, let $x$ be the one that minimizes the angle between $vx$ and $vu$. Note that $v$ cannot be the endpoint of a constraint intersecting the interior of triangle $uvx$, since the endpoint of that constraint would lie inside the triangle, contradicting our choice of $x$. Since both $uv$ and $vx$ are visibility edges, Lemma~\ref{lem:ConvexChain} tells us that there is a convex chain of visibility edges connecting $u$ and $x$ inside triangle $uvx$. In particular, the first vertex $y$ from $u$ on this chain is visible from both $u$ and $v$ and is closer to $u$ than $v$ is (in fact, $y$ must be $x$ by our choice of $x$). Moreover, $v$ must be in the same subcone of $u$ as $y$, since the region between $v$ and the chain is completely empty of both vertices and constraints. Thus, $uv$ cannot be an edge of the half-$\Theta_6$-graph.
 
 \begin{figure}[ht]
  \begin{center}
    \includegraphics{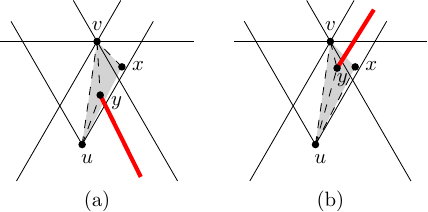}
  \end{center}
  \caption{(a) If $v$ is not closest to $u$ among the vertices visible to $v$, then $uv$ is not in the half-$\Theta_6$-graph. (b) If $v$ is closest to $u$ among the vertices visible to $v$, then $uv$ must be in the half-$\Theta_6$-graph.}
  \label{fig:LocalConstruction}
\end{figure}
 
Next, suppose that $v$ is closest to $u$ among the vertices visible to $v$ and not blocked by constraints incident on $v$, but $uv$ is not an edge of the half-$\Theta_6$-graph. Then there is a vertex $x \in C_0^u$ in the same subcone as $v$, who is visible to $u$, but not to $v$, and whose projection on the bisector is closer to $u$ (see Fig.~\ref{fig:LocalConstruction}b). Since $x$ and $v$ are in the same subcone, $u$ is not incident to any constraints that intersect the interior of triangle $uvx$. We now apply Lemma~\ref{lem:ConvexChain} to the triangle formed by visibility edges $uv$ and $ux$; this gives us that there is a convex chain of visibility edges connecting $v$ and $x$, inside triangle $uvx$. In particular, the first vertex $y$ from $v$ on this chain must be visible to both $u$ and $v$. And since $y$ lies in triangle $uvx$, it lies in $C_0^u$ and its projection is closer to $u$. But this contradicts our assumption that $v$ was the closest vertex. Thus, if $v$ is the closest vertex, and $uv$ must be an edge of the half-$\Theta_6$-graph.
\end{proof}

For completeness, we give the pseudocode of the complete procedure of determining whether a given edge in a negative (sub)cone is part of the constrained half-$\Theta_6$-graph in Algorithm~\ref{alg:GraphConstruction}. We note that it suffices to check using the two constraints that minimize the clockwise and counterclockwise angle with $uv$: If either constraint ends inside $\overline{C_0^v} \cap C_0^u$ it also serves as a witness that $uv$ is not an edge of the constrained half-$\Theta_6$-graph. Otherwise, neither ends in said region and they are the constraints that most restrict the visibility region of $u$ in $C_0^u$. Thus, they can be used to check the visibility and we do not have to test all constraints to do this. 

\begin{algorithm}
    \caption{Determining if $uv \in \overline{C_0^v}$ is part of the constrained half-$\Theta_6$-graph.}
    \begin{algorithmic}[1]
      \State Let \textsc{left} and \textsc{right} be the constraints incident to $v$ whose clockwise and counterclockwise angle with $uv$ is minimized
      \If{\textsc{left} or \textsc{right} in $C_0^u$}
        \State \Return \texttt{false}
      \Else
        \For{every edge $xv \in \overline{C_0^v} \cap C_0^u$}
          \If{$xv$ not blocked from $u$ by \textsc{left} or \textsc{right}}
            \State \Return \texttt{false}
          \EndIf
        \EndFor
      \EndIf
      \State \Return \texttt{true}
    \end{algorithmic}
    \label{alg:GraphConstruction}
\end{algorithm}

The running time of Algorithm~\ref{alg:GraphConstruction} depends highly on how the edges and constraints are stored at the vertices. If the edges and constraints are stored in arbitrary order, the procedure takes $O(C + \textrm{deg}(v))$ time, where $C$ is the number of constraints incident to $v$ and $\textrm{deg}(v)$ is the degree of $v$, since we need to loop over all of them in order to check the conditions. If the constraints are stored in clockwise or counterclockwise order around $v$, the time complexity is reduced to $O(\log C + \textrm{deg}(v))$, as we can use binary search to find the required constraints. If the edges are stored per subcone of $v$, we only need to loop through the edges in the subcone of $v$ that contains $u$, however in the worst case, this can still take $\textrm{deg}(v)$ time and the total time complexity remains $O(\log C + \textrm{deg}(v))$. 

Lemma \ref{lem:localnegedge} allows us to compute 1-locally which of the edges of $\Vis(P,S)$ incident on $v$ are also edges of the constrained half-$\Theta_6$-graph. Recall that this graph is plane~\cite{BFRV12Constrained}, thus we can apply FACE-1 or FACE-2 to route on $\Vis(P,S)$. The running time of the full routing procedure naturally depends on the way the edges and constraints are stored as well as whether FACE-1 or FACE-2 is used, and is the simple multiplication of the number of vertices visited and the total time spent per vertex. 

\begin{theorem}\label{theo_routing1}
  For any set $P$ of $n$ vertices and set $S$ of constraints on $P$, there exists a 1-local non-competitive routing algorithm on $\Vis(P,S)$ that visits only the edges of the constrained half-$\Theta_6$-graph. 
\end{theorem}

This algorithm routes on a subgraph of the constrained $\Theta_6$-graph, and in~\cite{BFRV2017RoutingJournal} it was shown that no deterministic local routing algorithm can be $o(\sqrt{n})$-competitive on this graph. Even worse, the competitive ratio of our approach cannot be bounded by any function of $n$, as it can traverse edges whose length is unrelated to (and thus much longer than) the shortest path. In fact, by applying FACE-1, it is possible to visit almost every edge of the graph four times before reaching the destination. It is worse with FACE-2, where almost every edge may be visited a linear number of times before reaching the destination. In the next section, we present a 1-local routing algorithm that is $O(n)$-competitive in the constrained setting and provide a matching worst-case lower bound.

\section{Routing on Constrained Triangulations}
In this section we look at routing on any constrained triangulation, i.e. a graph where all constraints are edges and all internal faces are triangles. Hence, we do not have to check while routing that the graph is a triangulation and we can focus our attention solely on the routing process. 

\subsection{Lower Bound}
\label{sec:lowerbound}
Given a triangulation $G$ and a source vertex $s$ and a destination vertex $t$, let $H$ be the subgraph of $G$ that contains all edges of $G$ that are part of a triangle that is intersected by $s t$. It is common for routing algorithms to restrict themselves to edges of $H$. Indeed, to the best of our knowledge, virtually every local routing strategy follows this restriction. In the unconstrained setting, this does not affect the quality of the path too much. For example, in the unconstrained Delaunay triangulation, $H$ always contains a path between $s$ and $t$ of length at most $2.42|st|$~\cite{KG92}. However, we show that this is no longer true in the constrained setting.

In particular, we show that if $G$ is a constrained Delaunay triangulation or a constrained half-$\Theta_6$-graph, the shortest path in $H$ can be a factor of $n/4$ times longer than that in $G$. This implies that any local routing algorithm that considers only the triangles intersected by $s t$ cannot be $o(n)$-competitive with respect to the shortest path in $G$ on every constrained Delaunay triangulation or constrained half-$\Theta_6$-graph on every pair of vertices. 
In the remainder of this paper, we use $\pi_G(u, v)$ to denote the shortest path from $u$ to $v$ in a graph $G$. 

\begin{lemma}
  There exists a constrained Delaunay triangulation $G$ with vertices $s$ and $t$ such that $|\pi_{H}(s, t)| \geq \frac{n}{4} \cdot |\pi_G(s, t)|$, where $H$ is the subgraph of $G$ consisting of all triangles intersected by the line segment $st$.
\end{lemma}
\begin{proof}
  In the following we show how to pick a set of points and constraints whose constrained Delaunay graph satisfies the properties stated in the Lemma. For ease of description and calculation, we assume that the size of the point set is a multiple of 4. Note that we can remove this restriction by adding 1, 2, or 3 vertices ``far enough away'' from the construction so that they do not influence the shortest path. 
  
  We start with two columns of $n/2 - 1$ vertices each, aligned on a grid. We add a constraint between every horizontal pair of vertices. Next, we shift every other row by slightly less than half a unit to the right (let $\varepsilon>0$ be the small amount that we did not shift). We also add a vertex $s$ below the lowest row and a vertex $t$ above the highest row, centered between the two vertices on said row. Note that this placement implies that $s t$ intersects every constraint. Finally, we stretch the point set by an arbitrary factor $2x$ in the horizontal direction, for some arbitrarily large constant $x$. When we construct the constrained Delaunay triangulation on this point set, we get the graph $G$ shown in Fig.~\ref{fig:LowerBound}. 

\begin{figure}[h]
  \begin{center}
    \includegraphics{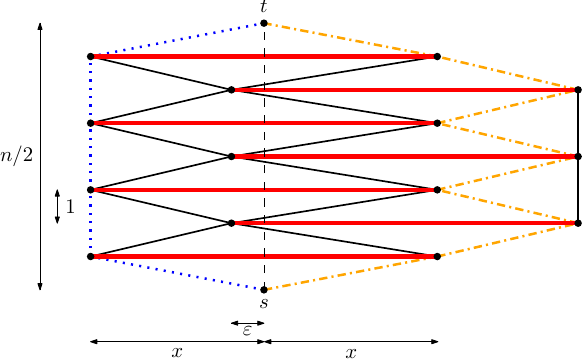}
  \end{center}
  \caption{Lower bound construction: the shortest path in $H$ (orange and dash-dotted) is about $n/4$ times as long as the shortest path in $G$ (blue and dotted). Constraints are shown in thick red and the remaining edges of $G$ are shown in solid black.}
  \label{fig:LowerBound}
\end{figure}

  In order to construct the graph $H$, we note that all edges that are part of $H$ lie on a face that has a constraint as an edge. In particular, $H$ does not contain any of the vertical edges on the left and right boundary of $G$. Hence, all that remains is to compare the length of the shortest path in $H$ to that in $G$. 
  
  Ignoring error terms up to $O(1)$, the shortest path in $H$ uses $n / 2$ edges of length $x$, hence it has length $x \cdot n/2$. Graph $G$ on the other hand contains a path of length $2 x + n/2 - 1$ (again, ignoring small terms that depend on $\varepsilon$), by following the path to the leftmost column and following the vertical path up. Hence, the ratio $|\pi_{H}(s, t)| / |\pi_G(s, t)|$ approaches $n/4$, since 
  $\lim_{x \rightarrow \infty} \frac{x \cdot \frac{n}{2}}{2 x + \frac{n}{2} - 1} = \frac{n}{4} \,.$
\end{proof}

Note that the above construction is also the constrained half-$\Theta_6$-graph of the given vertices and constraints.

\begin{corollary}
  There exist triangulations $G$ such that no local routing algorithm that considers only the triangles intersected by $s t$ is $o(n)$-competitive when routing from $s$ to $t$.
\end{corollary}

In fact, the construction depicted in Fig.~\ref{fig:LowerBound} shows that there exist point sets and constraints, such that the shortest path between $s$ and $t$ in \emph{every} triangulation on this point set (subject to the constraints) has length a linear factor shorter than the shortest path in $H$. 

\begin{lemma}
  There exist point sets $P$ (including vertices $s$ and $t$) and constraints $S$ such that in every constrained triangulation $G$ on $P$ subject to $S$, the shortest path between $s$ and $t$ in $H$ is not an $o(n)$-approximation of the shortest path in $G$, where $H$ is the subgraph of $G$ consisting of all triangles intersected by the line segment $st$.
\end{lemma}
\begin{proof}
Since any triangulation contains the edges of the convex hull of the point set, we observe that the shortest path in Fig.~\ref{fig:LowerBound} remains part of any triangulation. Hence, for $H$ to contain a shortest path of length comparable to the shortest path in the full triangulation, it needs to contain some vertical edge on the left or right boundary of the graph. We show that $H$ can contain no such edge. 

Consider an edge $u v$ on the vertical boundary of the triangulation and let $w$ be the third vertex of this triangle. Since in the construction the shifted constraints are shifted less than half a unit, the only vertices visible to both $u$ and $v$ are endpoints of a constraint whose $y$-coordinate lies between those of $u$ and $v$. Since $uvw$ is part of $H$, it intersects $st$, hence $w$ lies on the opposite side of $st$ compared to $u$ and $v$. This implies that the other endpoint of the constraint with endpoint $w$ is contained in $uvw$ and thus $uvw$ is not a triangle of the triangulation. 
\end{proof}

\subsection{Upper Bound}
Next, we provide a simple local routing algorithm that is $O(n)$-competitive. If we are only interested in routing on $H$, Bose and Morin~\cite{BM2004} introduced a routing algorithm for this setting. This routing algorithm, called the {\em Find-Short-Path} routing, is designed precisely to route on the graph created by the union of the triangles intersected by the line segment between the source and destination (i.e., graph $H$). The algorithm is 1-local and 9-competitive; that is, it reaches $t$ after having travelled at most 9 times the length of the shortest path from $s$ to $t$ in $H$, while considering only the neighbors of the current vertex. 

In the following, we show that this algorithm is also competitive in any triangulation. 

\begin{theorem}
  For any triangulation, there exists a 1-local $O(n)$-competitive routing algorithm that visits only triangles intersected by the line segment between the source and the destination. 
\end{theorem}

The remainder of the section is dedicated to showing that in any triangulation $G$ the shortest path between $s$ and $t$ in $H$ is an $O(n)$-approximation of the shortest path in $G$. To make the analysis easier, we use an auxiliary graph $H'$ defined as follows: let $H'$ be the graph $H$, augmented with the edges of the convex hull of $H$. These convex hull edges create additional faces, which we call {\em internal faces}. Notice that, by construction of $H$, each internal face can be associated with a unique edge of the convex hull. We call this edge the {\em defining edge}. Next, we add all visibility edges between vertices on the same internal face. For these visibility edges, we only consider constraints with both endpoints in $H$. The different graphs $G$, $H$, and $H'$ are shown in Fig.~\ref{fig:ConstructingH}. We emphasize that $H'$ is an auxiliary graph that will only be used to bound the spanning ratio between the other two graphs.

\begin{figure}[h]
  \begin{center}
    \includegraphics{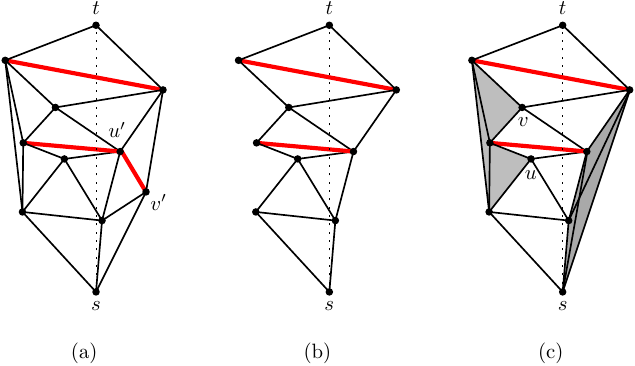}
  \end{center}
  \caption{The three different graphs: (a) The original triangulation $G$, (b) the subgraph $H$ containing only the triangles that intersect the segment $st$, (c) graph $H'$ constructed by adding convex hull edges to $H$ and visibility edges of the internal faces (gray regions in the figure). Note that edge $u v$ is not added, since visibility is blocked by a constraint that has both endpoints in $H$. Further note that in the right gray region we add ``illegal'' edges that cross the constraint $u'v'$.}
  \label{fig:ConstructingH}
\end{figure}

We start by comparing the length of the shortest paths in $H'$ and $G$. 

\begin{lemma}
\label{lem:VirtualShort}
  Any triangulation $G$ satisfies $|\pi_{H'}(s, t)| \leq |\pi_G(s, t)|$. 
\end{lemma}
\begin{proof}
  First consider the case where every vertex along $\pi_G(s, t)$ is part of $H'$. In this case, we claim that every edge of $\pi_G(s, t)$ is also part of $H'$. Clearly, if an edge $u v$ of $\pi_G(s, t)$ is part of a triangle intersected by $s t$, then it is included in $H$ (and therefore in $H'$). If $u v$ is not part of a triangle intersected by $s t$, then $u$ and $v$ must lie on the same internal face of $H'$ before we add the visibility edges (since otherwise the edge $u v$ would violate the planarity of $G$). Since $u v$ is an edge of $G$, $u$ and $v$ can see each other. Hence, the edge $u v$ is added to $H'$ when the visibility edges are added to the internal faces. Therefore, every edge of $\pi_G(s, t)$ is part of $H'$ and thus $|\pi_{H'}(s, t)| \leq |\pi_G(s, t)|$.
  
  In the general case not every vertex of $\pi_G(s, t)$ is part of $H'$. In this case we partition $\pi_G(s, t)$ into smaller subpaths so that each subpath satisfies either $(i)$ all vertices are in $H'$, or $(ii)$ only the first and last vertex of the subpath are in $H'$. Using an argument analogous to the previous case, it can be shown that subpaths $(u, v)$ of $\pi_G(s, t)$ that satisfy $(i)$ use only edges that are in $H'$ and thus that $|\pi_{H'}(u, v)| \leq |\pi_G(u, v)|$. 
  
  To complete the proof, it remains to show that given a subpath $\pi'$ that satisfies $(ii)$, there exists a different path in $H'$ that connects the two endpoints of $H'$ and has length at most $|\pi'|$. Let $u$ and $v$ be the first and last vertex of $\pi'$, and consider first the case where $u$ and $v$ lie on the same internal face (see Fig.~\ref{fig:PathNotInH}). $H'$ contains all visibility edges that are not blocked by constraints with {\em both} endpoints in $H'$. In particular, it will contain the geodesic $\pi_{H'}$ (i.e., the shortest possible path that avoids these constraints) between $u$ and $v$. On the other hand, path $\pi'$ uses only edges of $G$ which by definition do not cross any constraints of $S$. Hence, this implies in particular that $\pi'$ does not cross any constraint that has both endpoints in $H$ and we conclude that the path $\pi'$ cannot be shorter than $\pi_{H'}$.

  \begin{figure}[ht]
  \begin{minipage}[t]{0.45\linewidth}
    \begin{center}
      \includegraphics{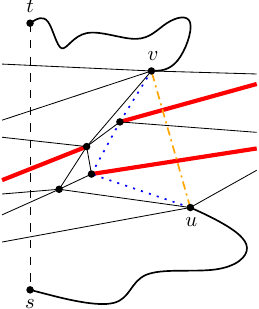}
    \end{center}
    \caption{A subpath of $\pi_G(s, t)$ (dotted blue) that satisfies condition $(ii)$: no vertex other than its endpoints are in $H'$. The two endpoints are connected in $H'$ (dot dashed orange path) and thus it has a shorter path in $H'$.} 
    \label{fig:PathNotInH}
  \end{minipage}
  \hspace{0.05\linewidth}
  \begin{minipage}[t]{0.45\linewidth}
    \begin{center}
      \includegraphics{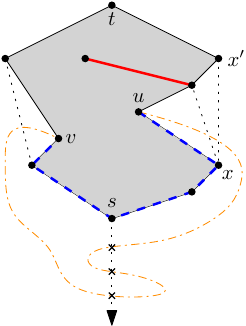}
    \end{center}
    \caption{When $\pi'$ (dot dashed orange) does not pass through any vertex of $H'$ (other than $u$ and $v$), we walk along the outer boundary of $H'$ to get a shorter path $\pi_{H'}$ (thick dashed blue). Note that we ignore some edges of $H'$ (dotted in the figure) in order to have $u$ and $v$ on the outer boundary.}
    \label{fig:SimulatingPath2}
  \end{minipage}
\end{figure}

  Finally, it remains to consider the case where $u$ and $v$ lie on different internal faces. Let $F$ be an internal face containing $u$.\footnote{Vertex $u$ can be in two internal faces if it is in the convex hull. In this case, we pick $F$ as either of the two faces.} Let $x$ and $x'$ be the endpoints of the defining edge of $F$. Consider the shortest path in $H'$ connecting $u$ with $x$ and $x'$ and virtually remove all edges from $F$ that do not belong to either path. We apply the same procedure to $v$. After this modification both $u$ and $v$ lie on the outer boundary of the pruned $H'$. We construct $\pi_{H'}$ by walking from $u$ to $v$ along this outer boundary. Note that there are two possible paths, clockwise or counterclockwise along the boundary; the path we choose will depend on $\pi'$.
  
  Without loss of generality, assume that $s$ is at the origin, $t=(0,1)$, and $u$ lies to the right of $s$ and $t$. We also assume that the clockwise path from $u$ to $v$ passes through $s$ (see Fig.~\ref{fig:SimulatingPath2}). We observe that since $\pi_G(s, t)$ is a shortest path in $G$ from $s$ to $t$, $\pi'$ is simple (i.e., no vertex is visited more than once). 
  
Next, consider $\pi'$ and recall that it satisfies $(ii)$ and thus no vertex along $\pi'$ other than $u$ and $v$ can be in $H'$. This implies that $\pi'$ cannot contain any vertex of the outer boundary of $H'$. We count the number times $\pi'$  crosses the downwards ray from $s$; if the number of crossings is odd, we construct $\pi_{H'}$ by walking clockwise from $u$ to $v$. Otherwise, we walk counterclockwise instead. Since we assumed that the clockwise path from $u$ to $v$ passes through~$s$ and $\pi'$ is simple, both $\pi'$ and $\pi_{H'}$ must have the same homotopy (if we virtually consider the outer boundary of $H'$ as an obstacle). Moreover, $\pi_{H'}$ is the shortest possible path having the same homotopy as $\pi'$. We conclude that $|\pi_{H'}| \leq |\pi'|$. 

Since for all subpaths $(u, v)$ of $\pi_G(s, t)$, we showed that $|\pi_{H'}(u, v)| \leq |\pi_G(u, v)|$, it follows that $|\pi_{H'}(s, t)| \leq |\pi_G(s, t)|$.  
\end{proof}

Next, we show that the length of the shortest path in $H$ has length at most $n - 1$ times the length of the shortest path in $H'$. 

\begin{lemma}
\label{lem:VirtualLinearBound}
  Any triangulation $G$ satisfies $|\pi_{H}(s, t)| \leq (n-1) \cdot |\pi_{H'}(s, t)|$. 
\end{lemma}
\begin{proof}
It suffices to show that every edge $uv$ on the shortest path in $H'$ can be replaced by a path in $H$ whose length is at most $|\pi_{H'}(s, t)|$. 
  The claim trivially holds if $u v$ is also an edge of $H$, thus we focus on the case where $u v$ is not an edge of $H$. Note that this implies that $u v$ is either an edge of the convex hull of $H$ or a visibility edge between two vertices of the same internal face. Instead of following $u v$, we \emph{simulate} $u v$ by following the path $\pi'$ along the pocket of $H$ from $u$ to $v$ (the path along the boundary of $H$ and the outer face that does not visit both sides of $s t$; see Fig.~\ref{fig:SimulatingPath}).

\begin{figure}[h]
  \begin{center}
    \includegraphics{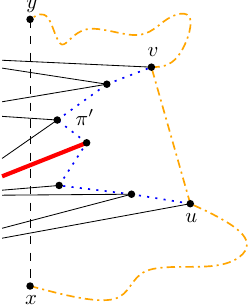}
  \end{center}
  \caption{The edge $uv$ on the shortest path in $H'$ (dot dashed orange) can be simulated with a path $\pi'$ (dotted blue) by walking along the face of a pocket of $H$. Any edge on that walk is contained in the polygon defined by the vertical segment $xy$ and the shortest path in $H'$.}
  \label{fig:SimulatingPath}
\end{figure}
  
We follow $\pi_{H'}(s, t)$ from $s$ to $t$ and consider the intersections between $\pi_{H'}(s, t)$ and the segment $st$ (they must cross at least twice: once at $s$ and once at $t$). Let $x$ be the last intersection before $u$ in $\pi_{H'}(s, t)$ and $y$ be the first intersection after $v$. Let $\mathcal{P}'$ be the polygon determined by segment $x y$, and the portion of $\pi_{H'}(s, t)$ that lies between $x$ and $y$. Since $\pi'$ lies on the boundary of a pocket, it cannot cross $s t$ and therefore it must be contained in $\mathcal{P}'$. In particular, all edges of $\pi'$ must lie inside $\mathcal{P}'$. Since a line segment inside a polygon has length at most half the perimeter of that polygon, the length of each edge of $\pi'$ is at most the length of $\pi_{H'}(s, t)$ from $x$ to $y$, which is at most $|\pi_{H'}(s, t)|$.

We concatenate all simulated paths and shortcut the resulting path from $s$ to $t$ such that every vertex is visited at most once. The result is a simple path which consists of at most $n-1$ edges, each one having length at most $|\pi_{H'}(s, t)|$. This completes the proof.
\end{proof}

By combining Lemmas~\ref{lem:VirtualShort} and \ref{lem:VirtualLinearBound} we obtain the desired ratio between the paths in $G$ and $H$.
\begin{theorem}
  Any triangulation $G$ satisfies $|\pi_{H}(s, t)| \leq (n-1) \cdot |\pi_G(s, t)|$. 
\end{theorem}

\section{Conclusions}
In this paper we presented two routing algorithms. The first one works in the visibility graph but its competitiveness is not bounded by any function of $n$ (the number of vertices). The second algorithm is $O(n)$-competitive (which is worst-case optimal due to the lower bound shown in Section~\ref{sec:lowerbound}), but it requires a triangulated subgraph of $\Vis(S,P)$. This leads to the following open problem: can one locally compute a triangulation of $\Vis(S,P)$? It is known that the constrained Delaunay triangulation cannot be computed locally (since it contains non-local information such as convex hull edges) and the constrained half-$\Theta_6$-graph is not necessarily a triangulation. Thus, finding a subgraph that can be computed locally from the visibility graph is a natural open problem. \\

\noindent\textbf{Acknowledgements}

\noindent We thank Luis Barba, Sangsub Kim, and Maria Saumell for fruitful discussions.

\bibliography{references}

\begin{thebibliography}{10}

\bibitem{BGHI10}
Nicolas Bonichon, Cyril Gavoille, Nicolas Hanusse, and David Ilcinkas.
\newblock Connections between theta-graphs, {D}elaunay triangulations, and
  orthogonal surfaces.
\newblock In {\em Proceedings of the 36th International Conference on Graph
  Theoretic Concepts in Computer Science (WG 2010)}, pages 266--278, 2010.

\bibitem{BFRV2017RoutingJournal}
Prosenjit Bose, Rolf Fagerberg, Andr\'e van Renssen, and Sander Verdonschot.
\newblock Competitive local routing with constraints.
\newblock {\em Journal of Computational Geometry (JoCG)}, 8(1):125--152, 2017.

\bibitem{BFRV12Constrained}
Prosenjit Bose, Rolf Fagerberg, Andr\'e van Renssen, and Sander Verdonschot.
\newblock On plane constrained bounded-degree spanners.
\newblock {\em Algorithmica}, 81(4):1392--1415, 2019.

\bibitem{BK06}
Prosenjit Bose and J.~Mark Keil.
\newblock On the stretch factor of the constrained {D}elaunay triangulation.
\newblock In {\em Proceedings of the 3rd International Symposium on Voronoi
  Diagrams in Science and Engineering (ISVD 2006)}, pages 25--31, 2006.

\bibitem{BKRV2017Routing}
Prosenjit Bose, Matias Korman, Andr\'e van Renssen, and Sander Verdonschot.
\newblock Constrained routing between non-visible vertices.
\newblock In {\em Proceedings of the 23rd Annual International Computing and
  Combinatorics Conference (COCOON 2017)}, volume 10392 of {\em Lecture Notes
  in Computer Science}, pages 62--74, 2017.

\bibitem{bkrv-rvg-17}
Prosenjit Bose, Matias Korman, Andr\'e van Renssen, and Sander Verdonschot.
\newblock Routing on the visibility graph.
\newblock {\em Journal of Computational Geometry (JoCG)}, 9(1):430–453, 2018.

\bibitem{BM2004}
Prosenjit Bose and Pat Morin.
\newblock Competitive online routing in geometric graphs.
\newblock {\em Theoretical Computer Science}, 324(2):273--288, 2004.

\bibitem{BMSU01}
Prosenjit Bose, Pat Morin, Ivan Stojmenovic, and Jorge Urrutia.
\newblock Routing with guaranteed delivery in ad hoc wireless networks.
\newblock {\em Wireless Networks}, 7(6):609--616, 2001.

\bibitem{BR14}
Prosenjit Bose and Andr\'e van Renssen.
\newblock Spanning properties of {Y}ao and $\theta$-graphs in the presence of
  constraints.
\newblock {\em International Journal of Computational Geometry \& Applications
  (IJCGA)}, 29(02):95--120, 2019.

\bibitem{C87}
Ken Clarkson.
\newblock Approximation algorithms for shortest path motion planning.
\newblock In {\em Proceedings of the 19th Annual ACM Symposium on Theory of
  Computing (STOC 1987)}, pages 56--65, 1987.

\bibitem{D97}
Gautam Das.
\newblock The visibility graph contains a bounded-degree spanner.
\newblock In {\em Proceedings of the 9th Canadian Conference on Computational
  Geometry (CCCG 1997)}, pages 70--75, 1997.

\bibitem{EM90}
Herbert Edelsbrunner and Ernst~Peter M{\"u}cke.
\newblock Simulation of simplicity: A technique to cope with degenerate cases
  in geometric algorithms.
\newblock {\em ACM Transactions on Graphics}, 9(1):66--104, 1990.

\bibitem{EC95}
Ioannis~Z. Emiris and John~F. Canny.
\newblock A general approach to removing degeneracies.
\newblock {\em SIAM Journal on Computing (SICOMP)}, 24(3):650--664, 1995.

\bibitem{KG92}
J.~Mark Keil and Carl~A. Gutwin.
\newblock Classes of graphs which approximate the complete euclidean graph.
\newblock {\em Discrete {\&} Computational Geometry}, 7:13--28, 1992.

\bibitem{KSU99}
Evangelos Kranakis, Harvinder Singh, and Jorge Urrutia.
\newblock Compass routing on geometric networks.
\newblock In {\em Proceedings of the 11th Canadian Conference on Computational
  Geometry (CCCG 1999)}, pages 51--54, 1999.

\bibitem{G09}
Sudip Misra, Subhas~Chandra Misra, and Isaac Woungang.
\newblock {\em Guide to Wireless Sensor Networks}.
\newblock Springer, 2009.

\bibitem{R09}
Harald R{\"a}cke.
\newblock Survey on oblivious routing strategies.
\newblock In {\em Mathematical {T}heory and {C}omputational {P}ractice}, volume
  5635 of {\em Lecture Notes in Computer Science}, pages 419--429, 2009.

\bibitem{Y90}
Chee-Keng Yap.
\newblock A geometric consistency theorem for a symbolic perturbation scheme.
\newblock {\em Journal of Computer and System Sciences}, 40(1):2 -- 18, 1990.

\end{thebibliography}

\end{document}